\documentclass[runningheads]{llncs}
\usepackage[T1]{fontenc}
\usepackage{graphicx}
\usepackage{hyperref} 

\usepackage{color}

\usepackage[firstpage]{draftwatermark} 

\SetWatermarkAngle{0}
\SetWatermarkText{\raisebox{13.0cm}{%
\hspace{0.1cm}%
\href{https://doi.org/10.5281/zenodo.10939233}{\includegraphics{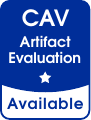}}%
\hspace{9cm}%
\includegraphics{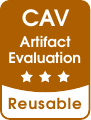}%
}}

\usepackage{amsthm,amsmath,stmaryrd,amsfonts,amssymb}
\usepackage{thmtools}
\usepackage{thm-restate}
\usepackage[conf={no link to proof,restate,text proof translated=}]{proof-at-the-end}

\usepackage{subcaption}
\usepackage{breqn}
\usepackage{tikz}
\usepackage{tikzscale}
\usepackage{pgfplots}
\usetikzlibrary{calc,matrix,arrows,automata,positioning,chains}
\usepackage{enumitem}
\usepackage{systeme}
\usepackage{array}

\newcommand{\tuple}[1]{\left\langle #1 \right\rangle}     
\newcommand{\set}[1]{\left\{ #1 \right\}}                 
\newcommand{\true}{\textit{true}}
\newcommand{\false}{\textit{false}}
\newcommand{\defeq}{\triangleq}

\newcommand{\floor}[1]{\left\lfloor#1\right\rfloor}

\newcommand{\conv}[1]{\textit{conv}\left(#1\right)}

\newcommand{\Qplus}{\mathbb{Q}^{\geq 0}}

\newcommand{\qideal}[1]{\left\langle #1 \right\rangle}
\newcommand{\cone}[1]{\textit{cone}(#1)}

\newcommand{\algpoly}{\textit{alg.polyhedron}\,}
\newcommand{\regcone}{\textit{alg.cone}\,}

\newcommand{\ThZ}{\mathbf{LIRR}}
\newcommand{\lirrmodels}{\models_{\ThZ}}

\newcommand{\CnX}[2]{\mathbf{C}_{#2}\!\left(#1\right)}

\newcommand{\LIA}{\textbf{LIA}}

\newcommand{\PRF}{\mathbf{PRF}}
\newcommand{\QPRF}{\mathbf{QPRF}}
\newcommand{\WLPRF}{\mathbf{WLPRF}}

\newcommand{\Int}{\textit{Int}\,}
\newcommand{\algorithmfont}[1]{\ensuremath{\texttt{#1}}\,}

\definecolor{black}{RGB}{0,0,0}
\definecolor{green}{RGB}{25,160,80}
\definecolor{red}{RGB}{192,57,43}
\definecolor{blue}{RGB}{41,128,185}
\definecolor{orange}{RGB}{255,99,0}
\definecolor{gray}{RGB}{75,75,75}
\definecolor{lightgrey}{RGB}{240,240,240}
\definecolor{purple}{RGB}{155, 89, 182}
\definecolor{darkgrey}{RGB}{52, 73, 94}

\usepackage[linesnumbered,noend]{algorithm2e}
\SetKwInOut{Input}{Input}
\SetKwInOut{Output}{Output}
\SetKwRepeat{Do}{do}{while}

\SetCommentSty{mycommfont}
\SetKwProg{Fn}{Function}{}{}
\SetKwProg{Sub}{Subroutine}{ begin}{end}
\SetKwComment{Inv}{\rm\textbf{invariant:} }{}

\usepackage{listings}
\lstdefinestyle{base}{
language=C,
emptylines=1,
breaklines=true,
basicstyle=\color{black}\fontfamily{pzc}\selectfont\tt,
commentstyle=\color{gray}\rm\itshape,
keywordstyle=\rm\bfseries,
identifierstyle=\rm\itshape,
escapechar=|,
morekeywords=[1]{then,do},
morekeywords=[2]{assert},
morekeywords=[3]{requires,ensures},
keywordstyle	= [2]\color{red}\bf,
keywordstyle	= [3]\bf\color{gray},
numberstyle=\footnotesize\color{gray},
moredelim=**[is][\bf\color{green}]{@!}{!@},
moredelim=**[is][\bf\color{red}]{@?}{?@},
literate=
	{<=}{{$\leq$}}1
{>=}{{$\geq$}}1
{!}{{$\neg$}}1
{!=}{{$\neq$}}1
{||}{{$\lor$}}1
{&&}{{$\land$}}1
{->}{{$\rightarrow$}}1
{_1}{$_{1}$}2
{_2}{$_{2}$}2
{_3}{$_{3}$}2
}

\DeclareFontFamily{U}  {MnSymbolC}{}
\DeclareFontShape{U}{MnSymbolC}{m}{n}{
	<-6>  MnSymbolC5
	<6-7>  MnSymbolC6
	<7-8>  MnSymbolC7
	<8-9>  MnSymbolC8
	<9-10> MnSymbolC9
	<10-12> MnSymbolC10
	<12->   MnSymbolC12}{}
\DeclareSymbolFont{MnSyC}{U}{MnSymbolC}{m}{n}
\DeclareMathSymbol{\righthalfcup}{\mathrel}{MnSyC}{184}
\newcommand{\romanqed}{$\righthalfcup$}
\newenvironment{mexample}[1][]{%
	\begin{example}%
		\pushQED{\let\qedsymbol\romanqed\qed}%
		}{%
		\popQED%
	\end{example}%
}
\newcommand{\romanqedhere}{\let\qedsymbol\romanqed\qedhere}

\usepackage{mathpartir} 
\usepackage{booktabs}   
\usepackage{hyperref}   
\usepackage[english]{babel}
\addto\extrasenglish{
	
}
\usepackage{soul}

\usepackage{tcolorbox}
\usepackage{mathtools}
\usepackage{graphics}
\usepackage{multirow}

\newenvironment{mproofEnd}[1][]{%
   \begin{proofEnd}%
 }{%
 \end{proofEnd}%
}

\begin{document}
\title{Breaking the Mold: Nonlinear Ranking Function Synthesis without Templates}
\titlerunning{Nonlinear Ranking Function Synthesis}
%
%
\author{Shaowei Zhu\orcidID{0000-0002-0335-1151}
\and Zachary Kincaid\orcidID{0000-0002-7294-9165}
}
\authorrunning{S. Zhu and Z. Kincaid}
 \institute{Princeton University, Princeton NJ 08540, USA\\
 \email{\{shaoweiz,zkincaid\}@cs.princeton.edu}}

\maketitle              

\begin{abstract}
This paper studies the problem of synthesizing (lexicographic) polynomial ranking functions for loops that can be described in polynomial arithmetic over integers and reals.  While the analogous ranking function synthesis problem for \emph{linear} arithmetic is decidable, even checking whether a \emph{given} function ranks an integer loop is undecidable in the nonlinear setting.
We side-step the decidability barrier by working within the theory of linear integer/real rings (LIRR) rather than the standard model of arithmetic.  We develop a termination analysis that is guaranteed to succeed if a loop (expressed as a formula) admits a (lexicographic) polynomial ranking function.  In contrast to template-based ranking function synthesis in \emph{real} arithmetic, our completeness result holds for lexicographic ranking functions of unbounded dimension and degree,
and effectively subsumes linear lexicographic ranking function synthesis for linear \emph{integer} loops.

\end{abstract}
\section{Introduction} 
\label{sec:introduction}

Ranking function synthesis refers to the problem of finding a well-founded metric that 
decreases at each iteration of a loop.  It is a critical subroutine in modern 
termination analyzers like Terminator \cite{cookTerminationProofsSystems2006}, 
Ultimate Automizer \cite{heizmannRefinementTraceAbstraction2009}, and ComPACT \cite{zhuTerminationAnalysisTears2021}.
One could synthesize ranking functions via a \emph{template}, i.e.,
fixing a particular form of ranking functions to be considered while
leaving parameters as free variables, and encoding the conditions 
for the function to rank the given loop as a logical formula, 
thereby reducing the synthesis problem to a constraint-solving problem. 
Provided that the resulting constraint-solving
problem is decidable, this method yields a complete procedure for
synthesizing ranking functions that match the template.  In particular, the
template-based method is the basis of complete synthesis of ranking
functions for linear and lexicographic linear ranking functions for loops whose
bodies and guards can be expressed in linear real or integer arithmetic \cite{podelskiCompleteMethodSynthesis2004,ben-amramRankingFunctionsLinearConstraint2014}.  
A limitation of
the approach is that it is only complete with respect to template languages that
can be defined by finitely many parameters (e.g., we may define a template for
all linear terms or degree-2 polynomials, but not polynomials of unbounded degree).
\footnote{One may imagine using the template paradigm to search for polynomial ranking
functions of successively higher degree until one is found; however, this yields
a complete \textit{semi-algorithm}, which fails to terminate if no polynomial ranking function exists.}


In this paper, we study the problem of synthesizing polynomial ranking functions for nonlinear loops.  There are two apparent obstacles.
The first obstacle results from the difficulty of reasoning about nonlinear arithmetic. 
Nonlinear \textit{integer} arithmetic is
undecidable, and so even checking whether a \textit{given} function ranks a loop
is undecidable, let alone synthesizing one.
While nonlinear real arithmetic is decidable, it has high complexity--prior work has explored \textit{incomplete} constraint-solving approaches to avoid the cost of decision procedures for real arithmetic \cite{cousotProvingProgramInvariance2005,PLDI:ACFGM2021}, but this sacrifices the completeness property typically enjoyed by template-based methods.
The second obstacle is that the set of all polynomials cannot be described as a template language with finitely many parameters, thus precluding complete ranking function synthesis based on the template method.

To tackle the undecidability problem, we adopt a weak theory of
nonlinear arithmetic $\ThZ$ that is
decidable \cite{kincaidWhenLessMore2023}. For the infinite template problem, 
we first compute the finite set of polynomials that are entailed to be \emph{bounded} modulo $\ThZ$
by the loop, and use them to define a template language with 
finitely many parameters to describe ``candidate terms'' for ranking functions. 
We then show that synthesis of
ranking functions consisting of non-negative linear combinations of these candidate terms
can be reduced to a constraint-solving problem in linear arithmetic.
The adoption of $\ThZ$ ensures that we do not lose completeness in any of the above steps, i.e.,
any ranking function modulo $\ThZ$ can be written as a nonnegative combination 
of the ``candidate terms'' in the template.
We thus have a procedure for synthesizing polynomial ranking functions that
is sound for the reals, and \textit{complete} in the sense that 
if a polynomial ranking function exists for a
formula (modulo $\ThZ$), then the analysis will find it.
Furthermore, we extend this analysis to one that is sound for the integers and
complete relative to lexicographic polynomial ranking functions (modulo $\ThZ$).

Using the framework of
algebraic termination analysis \cite{zhuTerminationAnalysisTears2021},
we extend our termination analysis on loops (represents as formulas) to whole programs
(including nested loops, recursive procedures, etc).
The completeness of the proposed procedures leads to \emph{monotone} 
end-to-end termination analyses for whole programs.
Informally, monotonicity guarantees that if the analysis can prove termination of a program $P$ and $P$ is transformed to a program $P'$ in a way that provides more information about its behavior
(e.g., by decorating the program with invariants discovered by an abstract interpreter)
then the analysis is certain to prove termination of $P'$ as well.


Our experimental evaluation establishes that the procedure based on polynomial
ranking function and lexicographic polynomial ranking function synthesis with
the background theory of $\ThZ$ is competitive for SV-COMP termination
benchmarks, especially for the nonlinear programs.


\section{Background} 
\label{sec:background}

\subsection*{Linear Algebra and Polyhedral Theory}

In the following, we use \textbf{linear space} to mean a linear space over the
field of rationals $\mathbb{Q}$. Let $L$ be a linear space. A set $C \subseteq
L$ is \textbf{convex} if for every $p, q \in C$ and every $\lambda \in [0, 1]$,
we have $\lambda p + (1-\lambda) q \in C$. We use $\conv{S}$ to denote the
\textbf{convex hull} of a set $S \subseteq L$, which is the smallest convex set
that contains $S$. A set $Q$ is a \textbf{polytope} if it is the
convex hull of a finite set.  A set $C \subseteq L$ is a \textbf{(convex) cone}
if it contains $0$ and is closed under addition and multiplication by $\Qplus$
(nonnegative rationals). For a set $G \subseteq L$, its \textbf{conical hull} is the
smallest cone that contains $G$, defined as $\cone{G} = \set{\lambda_1 g_1 +
\dots + \lambda_m g_m: \lambda_i \in \Qplus, g_i \in G}$. Given any $A, B
\subseteq L$, we use $A+B \defeq \set{ a + b : a \in A, b \in B}$ to denote
their Minkowski sum. 

A set $P \subseteq L$ is a \textbf{polyhedron} if $P = \cone{R} + \conv{V}$,
where $R, V$ are finite sets in $L$, and use the notation $P = \textit{V-rep}(R,
V)$. Convex polyhedra are effectively closed under intersection; that is, there
is a procedure \algorithmfont{intersect} such that for any finite
$R_1,V_1,R_2,V_2 \subseteq L$ we have
\[
\textit{V-rep}(\algorithmfont{intersect}(R_1,V_1,R_2,V_2)) =
\textit{V-rep}(R_1,V_1) \cap \textit{V-rep}(R_2,V_2) \ .
\]

\subsection*{The Ring of Rational Polynomials}


For a finite set of variables $X$, we use $\mathbb{Q}[X]$ to denote the ring of polynomials over $X$ with rational coefficients, and $\mathbb{Q}[X]^1$ to denote the set of linear polynomials over $X$.
A set $I \subseteq \mathbb{Q}[X]$ is an 
\textbf{ideal} if it contains zero, is closed under addition, and for every $p \in \mathbb{Q}[X]$ and $q \in I$ we have $pq \in I$.
For a finite set $G = \set{g_1,\dots,g_n} \subseteq \mathbb{Q}[X]$, we use $\qideal{G} \defeq \set{ p_1g_1 + \dotsi + p_ng_n : p_1,\dots,p_n \in \mathbb{Q}[X]}$ to denote
the ideal generated by the elements in $G$.  By Hilbert's basis theorem, we have that every ideal in $\mathbb{Q}[X]$ can be written as $\qideal{G}$ for some finite set $G$.
Equivalently, for any ascending chain of ideals $I_1 \subseteq I_2 \subseteq \dots $ in
$\mathbb{Q}[X]$,  there exists an index $j$ such that $I_j = I_k$ for all $k \geq j$. 

Note that $\mathbb{Q}[X]$ is a linear space over $\mathbb{Q}$, and so cones, polytopes, and polyhedra consisting of polynomials are defined as above.  
We say that a cone $C \subseteq \mathbb{Q}[X]$ is \textbf{algebraic} if it is the Minkowski sum of an ideal and a finitely-generated convex cone \cite{kincaidWhenLessMore2023}.  For finite sets of polynomials $Z,P \subseteq \mathbb{Q}[X]$, we use
\[
	\regcone_X(Z, P) \defeq \set{ \sum_{z \in Z}  q_z z +
\sum_{p \in P} \lambda_p p : q_z \in \mathbb{Q}[X], \lambda_p \in \Qplus}
\]
to denote the algebraic cone generated by $Z$ and $P$; we call $Z$ and $P$ the ``zeros'' and ``positives'' of the cone, respectively.
When the set of variables is clear,
we often omit the subscript and just write $\regcone(Z, P)$.



For any
algebraic cone $C \subseteq \mathbb{Q}[X]$, the set of \textit{linear}
polynomials in $C$ forms a convex polyhedron.  We use
$\algorithmfont{linearize}$ to denote the operation that computes this
set---that is, for any finite $Z,P \subseteq \mathbb{Q}[X]$, we have 
\[\textit{V-rep}(\algorithmfont{linearize}(Z,P)) = \regcone(Z,P) \cap \mathbb{Q}[X]^1 \ .\]

There is a
procedure $\algorithmfont{inverse-hom}$ for computing the inverse image of an algebraic cone
under a ring homomorphism (\cite{kincaidWhenLessMore2023}, Theorem 9).
More precisely, let $\regcone_X(Z, P)$ be an algebraic cone, $Y$ be
a set of variables, and $f:\mathbb{Q}[Y] \rightarrow
\mathbb{Q}[X]$ be a ring homomorphism, then 
\[
	\regcone_Y(\algorithmfont{inverse-hom}(Z, P, f, Y)) = \set{p \in \mathbb{Q}[Y]: f(p) \in \regcone_X(Z, P)} \ .
\] 

In this paper it will be useful to define a common generalization algebraic
cones and convex polyhedra, which we call a \textit{algebraic polyhedra}. We say
that a set of polynomials $R \subseteq \mathbb{Q}[X]$ is an \textbf{algebraic polyhedron}
if it is the Minkowski sum of an algebraic cone and a convex
polytope\footnote{Recalling that a \textit{convex} polyhedron is the Minkowski
sum of a \textit{finitely generated convex} cone and a polytope.}.  An algebraic
polyhedron can be represented by a triple $\tuple{Z,P,V}$ where $Z,P,V$ are
finite sets of polynomials; such a triple represents the algebraic polyhedron
\[ \algpoly(Z, P, V) \defeq \regcone(Z, P) + \conv{V}\ . \]

\subsection*{The Arithmetic theory LIRR and Consequence Finding}

We use the following syntax for formulas:
\begin{align*}
 F,G \in \textbf{Formula} &::= p \leq q \mid p = q \mid \textit{Int}(p) \mid F \land G \mid F \lor G \mid \lnot F
\end{align*}
where $p$ and $q$ denote polynomials with rational coefficients over some set of variable symbols.
We regard the reals $\mathbb{R}$ as the standard interpretation of this language, with $\textit{Int}$ identifying the subset of integers $\mathbb{Z} \subset \mathbb{R}$.

Kincaid et al. \cite{kincaidWhenLessMore2023} defined another class of interpretations for the above language of formulas called \textit{linear integer/real rings}.  A linear integer/real ring is a commutative ring equipped with an order and an integer predicate which obeys certain axioms of the theories of linear real and linear integer arithmetic.  The standard interpretation $\mathbb{R}$ is an example of a linear integer/real ring.  A ``nonstandard'' example is the ring $\mathbb{Q}[x]$, where $p \leq q$ iff $p$ precedes $q$ lexicographically (e.g., $-x^3 < x < x^2 - x < x^2 < x^2 + x $) and $\textit{Int}(p)$ holds iff $p$'s coefficients are integers. 

The fact that the theory $\ThZ$ of linear integer/real rings (refer to \cite{kincaidWhenLessMore2023} for an axiomatization) admits such nonstandard (and inequivalent) models means that the \textit{theory} is incomplete. Nevertheless it has desirable algorithmic properties that we will make use of in our ranking function synthesis procedures. We discuss
the limitations brought by $\ThZ$ in \autoref{eg:incompleteness-prf-due-to-lirr}.

Since the reals $\mathbb{R}$ is a model for $\ThZ$, if we have $F
\lirrmodels G$, we also have $F \models_{\mathbb{R}} G$. However, in this paper
we are mostly concerned with entailment modulo $ \ThZ $ rather than the
standard model, thus we abbreviate $F \lirrmodels G$ to $ F \models G $ by
default.

For a formula $F$ and a set of variables $X$, we use
\[\CnX{F}{X} \defeq \set{ p \in \mathbb{Q}[X] : F \models p \geq 0 }\]
to denote the \textbf{nonnegative cone} of $F$ (over $X$). 
For example, given $X = \set{x, y}$  
\[
\CnX{x = 2 \land y \leq 1}{X} = \regcone(\set{x - 2}, \set{1, 1 - y}) \ . 	
\]
$\CnX{F}{X}$ is an algebraic cone, and there is an algorithm for computing it (Algorithm 2 of \cite{kincaidWhenLessMore2023}), which we denote by $\algorithmfont{consequence}(F, X)$.  We furthermore have that if $\tuple{Z,P} = \algorithmfont{consequence}(F, X)$, then
$\qideal{Z} = \set{ z \in \mathbb{Q}[X] : F \models z = 0 }$.



\subsection*{Transition Systems and Transition Formulas}


For a set of variables $X$, we use $X' \defeq \set{ x' : x \in X}$ denote a set
of ``primed copies''.  For a polynomial $p \in \mathbb{Q}[X]$, we use $p'$ to
denote the polynomial in $\mathbb{Q}[X']$ obtained by replacing each variable
$x$ with its primed copy $x'$. A \textbf{transition formula} over a set of
variables $X$ is a formula $F$ whose free variables range over
$X$ and $X'$.  We use $\mathbf{TF}(X)$ to denote the set of all transition
formulas over $X$. For a transition formula $F \in \mathbf{TF}(X)$ and real
valuation $v,v' \in \mathbb{R}^X$, we use $v \rightarrow_F v'$ to denote that
$\mathbb{R},[v,v'] \models F$, where $\mathbb{R}$ denotes the standard model and $[v,v']$ denotes the valuation that maps each $x \in
X$ to $v(x)$ and each $x' \in X'$ to $v'(x)$.  
A \textbf{real execution} of a transition formula $F$ is an infinite sequence
$v_0, v_1, \dots \in \mathbb{R}^X$ such that for each $i$, we have $v_i
\rightarrow_F v_{i+1}$; we say that $v_0,v_1,\dots$ is an \textbf{integer
execution} if additionally each $v_i \in \mathbb{Z}^X$.  We say that \textit{$F$
terminates over $\mathbb{R}$} if it has no real executions, and \textit{$F$
terminates over $\mathbb{Z}$} if it has no integer executions.


\subsection*{Ranking Functions}

Let $F \in \mathbf{TF}(X)$ be a transition formula.  
We say that $r \in \mathbb{Q}[X]$ is a  \textbf{polynomial
ranking function} ($\PRF$) for $F$ (modulo $\ThZ$) if $F \models 0 \leq r$ and $F
\models r' \leq r - 1$. The set of all polynomial ranking functions of $F$
(modulo $\ThZ$) is denoted $\PRF(F)$.

\begin{lemma}
	\label{lem:prf-implies-termination-over-Q}
	If $\PRF(F) \neq \emptyset$, then $F$ terminates over $\mathbb{R}$.
\end{lemma}
\begin{proof}
	If $r \in \PRF(F)$, then 
$\floor{r(X)}$ is a ranking function mapping $\mathbb{R}^X$ into $\mathbb{Z}$ that is
well-ordered by a relation $\preceq$, defined as  $x \preceq y$ iff $x \geq 0
\land x \leq y$, where $\leq$ is the usual order on the integers.
\end{proof}

We now consider lexicographic termination arguments.
We define a \emph{quasi-polynomial ranking function} ($ \QPRF $) for a
 transition formula $F \in \mathbf{TF}(X)$ (modulo $\ThZ$) to
be a polynomial $r \in \mathbb{Q}[X]$ such that
\[
	F \models r - r' \geq 0 \land r \geq 0 \ .
\]
We say that a sequence of
polynomials $r_1,\dots, r_n \in \mathbb{Q}[X]$ is a dimension-$n$ \textbf{weak
lexicographic polynomial ranking function} ($\WLPRF$) for $F$ (modulo $\ThZ$) if
\begin{align*}
    &r_1 \in \QPRF(F) \\ 
    &r_2 \in \QPRF(F \land r_1'=r_1) \\
    &\vdots \\ 
    &r_n \in \QPRF\left(F \land \bigwedge_{i=1}^{n-1} r_i' = r_i\right) \\
    &F \land \bigwedge_{i=1}^{n} r_i' = r_i \models \false \ .
\end{align*}

\autoref{lem:wlprf-implies-termination-over-Z} sketches the proof that the
existence of $\WLPRF$ proves termination of $F$ over $\mathbb{Z}$.

\begin{lemma} \label{lem:quasi-rf-implies-termination}
    Let $F \in \mathbf{TF}(X)$ be a transition formula. If $r \in \mathbb{Q}[X]$ is a quasi-ranking function for $F$,
    i.e., $F \models r' \leq r \land r \geq 0$, and furthermore $F \land r' = r $ terminates over the integers,
    then so does $F$.
\end{lemma}
\begin{proof}
  Since quasi-ranking functions are closed under scaling by nonnegative scalars,
  we may assume that $r$ has integer coefficients without loss of generality.
    Suppose for a contradiction that $F$ has an infinite
    integer execution $x_0, x_1, \dots$. Since $r(x_i) \geq r(x_{i+1})$ for all $i$, 
    and the range of $r$ is restricted to $\mathbb{Z}^{\geq 0}$,
    there exists some $n$ such that $r(x_{n}) = r(x_{n+1}) = \dots$. But this is impossible since 
    $F \land r' = r$ terminates over the integers.
\end{proof}

\begin{lemma}
	\label{lem:wlprf-implies-termination-over-Z}
	If a transition formula $ F $ admits a $ \WLPRF $ (modulo the theory $\ThZ$),
	then $F$ terminates over $ \mathbb{Z} $.
\end{lemma}
\begin{proof}
 We prove this by induction on the dimension $ n $ of $ \WLPRF $ of $ F $.
 The base case holds vacuously when $n = 0$ since $F$ is unsatisfiable,
 and the inductive case holds by \autoref{lem:quasi-rf-implies-termination}.
\end{proof}

Note that \autoref{lem:wlprf-implies-termination-over-Z} holds only for integer executions. Ben-Amram and Genaim \cite{ben-amramRankingFunctionsLinearConstraint2014} showed
that existence of a weak lexicographic linear ranking function (LLRF) for a 
topologically closed linear formula implies
existence of an LLRF for loop with real variables, 
but the argument fails for \emph{nonlinear} formulas (even
modulo $\ThZ$). Consider the following $ \ThZ $ transition formula over reals
$n, z$
\[
	F \defeq z \geq 0 \land n \geq 2 \land n' = 2 n \land z \geq z' \land nz' = nz - 1\ .
\]
Then $ F \models z \geq 0 \land z \geq z'  $, and also $ F \land z = z' \models
\false $. Thus $z$ does decrease at every iteration of $F$ and its value is
bounded from below. However, $ F $ does not terminate since the rate at which
$z$ decreases diminishes too quickly. 




\section{Polynomial Ranking for $\ThZ$ Transition Formulas} 
\label{sec:prf}

In this section, we consider the problem of synthesizing polynomial ranking
functions for transition formulas modulo $\ThZ$.
Observe that for a transition formula $F \in \mathbf{TF}(X)$, the polynomial ranking functions $\PRF(F)$ of $F$ can be decomposed as
$\PRF(F) = \textit{Bounded}(F) \cap \textit{Decreasing}(F)$
where $\textit{Bounded}(F)$ are the bounded and decreasing polynomials of $F$, respectively:
\begin{align*}
    \textit{Bounded}(F) &\defeq \set{ p \in \mathbb{Q}[X] : F \models p \geq 0}\\
    \textit{Decreasing}(F) &\defeq \set{ p \in \mathbb{Q}[X] : F \models p' \leq p - 1}
\end{align*}
Thus, one approach to computing $\PRF(F)$ is to compute the sets of bounded and decreasing polynomials, and then take the intersection.

First, we observe that we can use this strategy to synthesize \textit{linear} ranking functions using 
the primitives defined in Section~\ref{sec:background}\footnote{This is essentially a recasting of the classic algorithms linear ranking function synthesis \cite{ben-amramRankingFunctionsLinearConstraint2014} for LRA, restated in our language.}.
\begin{itemize}
\item The convex polyhedron of degree-1 polynomials of 
$\textit{Bounded}(F)$ can be computed as $\texttt{linearize}(\texttt{consequence}(F,X))$,
\item The convex polyhedron of degree-1 polynomials of
$\textit{Decreasing}(F)$ can be computed as follows.
Define $f : \mathbb{Q}[X] \rightarrow \mathbb{Q}[X \cup X']$ to be the homomorphism mapping $x \mapsto x-x'$, and observe that
\[
\textit{Decreasing}(F) \cap \mathbb{Q}[X]^1
= \set{ p \in \mathbb{Q}[X]^1 : F \models f(p) - 1 \geq 0}
\]
We proceed by first computing the polyhedron
\[ Q \defeq \set{ p + a : p \in \mathbb{Q}[X]^1, a \in \mathbb{Q}. F \models f(p) + a \geq 0} \]
as $\texttt{linearize}(\texttt{inverse-hom}(\texttt{consequence}(F,X \cup X'), f))$.
Then we intersect $Q$ with the hyperplane consisting of linear polynomials with constant coefficient -1,
 and then take the Minkowski sum with the singleton $\set{1}$ to get $\textit{Decreasing}(F) \cap \mathbb{Q}[X]^1$.
\end{itemize}

The essential difficulty of adapting this strategy to find polynomial ranking functions of unbounded degree is that the function $g : \mathbb{Q}[X] \rightarrow \mathbb{Q}[X \cup X']$ mapping $p \mapsto p'-p$ is not a homomorphism (the function $f$ defined above agrees with $g$ on \textit{linear} polynomials, but not on polynomials of greater degree).

Our method proceeds as follows. As we will later see in \autoref{alg:prf}, we can adapt the above strategy to compute the intersection of 
$\PRF(F)$ with some ``template language'' $\set{ a_1p_1 + \dots + a_np_n : a_1,\dots,a_n \in \mathbb{Q} }$ for fixed polynomials $p_1,\dots,p_n$.  Our insight is to use
the cone generators of $\textit{Bounded}(F)$ to define $p_1,\dots,p_n$.  This yields a ranking function synthesis procedure that, in general, is sound but incomplete;
however, it is complete under the assumption that $F$ is \textit{zero-stable}.  In Section~\ref{sec:zero-stable-restrictions} we define zero-stability and show that assuming zero-stability is essentially without loss of generality, and in Section~\ref{sec:prf} we define a procedure for computing $\PRF(F)$ for zero-stable $F$.

\subsection{Zero-Stable Transition Formulas}
\label{sec:zero-stable-restrictions}

Consider a transition formula $F$ defined as
\[ F \defeq x = 0 \land y \geq 0 \land (x')^2 = y - y' - 1 \ . \] 
Observe that $F \models x = 0$.  $F$ has a PRF $x^2 + y$, but it's hard to find in the
sense that it's not a linear combination of the generators of
$\textit{Bounded}(F)$ ($x$, $-x$, and $y$).  But when $x' \neq 0$, the loop terminates
immediately. Thus we can consider the restriction $F \land x' = 0$, which
admits the linear ranking function $y$. The zero-stable restriction process we
introduce below formalizes this process.

We define a transition formula $ F \in \mathbf{TF}(X) $ to be \textbf{zero-stable}
if for all polynomials $ p \in \mathbb{Q}[X] $ such that $ F \models p = 0 $, it
is the case that $ F \models p' = 0 $.
We give an algorithm for computing the
weakest zero-stable transition formula that entails the original formula in
\autoref{alg:zero-stable-restrict}, and we note that the algorithm
preserves termination behavior (\autoref{lem:zero-stable-restrict-properties}).




\begin{algorithm}
	\caption{The zero-stable restriction of a transition formula. \label{alg:zero-stable-restrict}}
	\SetKwFunction{Fconseq}{consequence}
	\SetKwFunction{Ftermprf}{terminates-WPRF}
	\SetKwFunction{Fzerostable}{zero-stable-restrict}
	\SetKwFunction{Fprf}{prf-zero-stable}
	\SetKwFunction{Fsat}{sat-lirr}
	\Fn{\Fzerostable{$F$}}{

		\KwIn{A transition formula $F \in \mathbf{TF}(X)$. }
		\KwOut{The weakest zero-stable transition formula that entails $F$. }
		$\tuple{Z, \_} \gets $\Fconseq{F,X}\;  \tcc{$\qideal{Z} = \set{p \in \mathbb{Q}[X] : F \models p = 0}$}
		\Repeat{$\qideal{Z} =\qideal{Z'}$}{
			$Z' \gets Z$\;
			$F \gets F \land \bigwedge_{z \in Z} z' = 0$\;
			$\tuple{Z, \_} \gets $ \Fconseq{F,X}\;
		}
		\Return{F}
	}
\end{algorithm}

\begin{lemma}
	\label{lem:zero-stable-restrict-properties}
	Let $F$ be an $\ThZ$ transition formula, and
	\[
		\hat{F} \defeq \algorithmfont{zero-stable-restrict}(F) \ .
	\]
	\begin{enumerate}
		\item \autoref{alg:zero-stable-restrict} computes the weakest zero-stable formula that entails $F$.
		\item $ F $ terminates iff $ \hat{F} $ terminates.
	\end{enumerate}
\end{lemma}
\begin{proof}
	Let $F^{(k)}$ and $Z^{(k)}$ denote the value of $F$ and $Z$ after the $k$-th iteration of the loop in
	\autoref{alg:zero-stable-restrict}, respectively.

	We first prove 1. Clearly, if \autoref{alg:zero-stable-restrict} terminates at some iteration $n$, then 
  $\hat{F} = F^{(n)}$ is zero stable and entails $F$.  It remains to show that (a) $F^{(n)}$ is the \textit{weakest} such formula, and (b) the algorithm terminates.
		      \begin{enumerate}[label=(\alph*)]
			      \item We show by induction that for any zero-stable formula $G$ that entails $F$,
			            it is the case that $G \models F^{(k)}$ for all $k$.
			            The base case holds by assumption, since $G \models F = F^{(0)}$.
			            Now suppose that $G \models F^{(k)}$, and we wish to show that
			            $G \models F^{(k+1)}$. Since for each $z \in Z^{(k)}$ we have
               $G \models F^{(k)} \models z = 0$, and $G$ is zero-stable, we know $G \models z' = 0$.  It follows $G \models (F^{(k)} \land \bigwedge_{z \in Z^{(k)}} z' = 0) = F^{(k+1)}$.
			      \item We show that \autoref{alg:zero-stable-restrict} terminates.
			            Suppose that it does not.
               Then $\qideal{Z^{(0)}} \subsetneq \qideal{Z^{(1)}} \subsetneq \dotsi$ forms an infinite strictly ascending chain of ideals of $\mathbb{Q}[X]$, contradicting Hilbert's basis theorem.
		      \end{enumerate}
	For 2, if $F$ terminates then $\hat{F}$ clearly also terminates since $\hat{F} \models F$.  To show that if $\hat{F}$ terminates then $F$ must also terminate, we 
			            prove by induction that for any $k \geq 0$, $
			            F^{(k+1)} $ terminates implies that $ F^{(k)} $
			            terminates. We show this by arguing that any real execution of $ F^{(k)} $ is also
			            one of $ F^{(k+1)} $. 
               Let $v_0, v_1, \dots $ be an execution of $F^{(k)}$.  It is sufficient to show that
               $v_{i} \rightarrow_{F^{(k+1)}} v_{i+1}$ for all $i$.  Since $v_{i+1} \rightarrow_{F^{(k)}} v_{i+2}$, we must have
               $z(v_{i+1}) = 0$ for all $z \in Z^{(k)}$,
               and so since $F^{(k+1)} = F^{(k)} \land \bigwedge_{z \in Z^{(k)}} z' = 0$ we have $v_{i} \rightarrow_{F^{(k+1)}} v_{i+1}$.
\end{proof}

\begin{mexample}
	Consider running \autoref{alg:zero-stable-restrict} on a transition formula $F$
	\[
		F: x = 0 \land y \geq 0 \land y' = -(x')^2 + y - 1 + z' \land z = x' \ .
	\]
	In the first iteration of the loop, we discover a zero consequence $x$ of
	$F$: $F \models x = 0$, and we then constrain the transition formula to be
	$F \land x' = 0$. Now since $F \models z = x'$, we get a new zero consequence $z$: $F
	\land x' = 0 \models z = 0 $. We thus further constrain the transition
	formula to be $F \land x' = 0 \land z' = 0 $. After adding these
	constraints, we can no longer find new zero consequences, and the resulting
	transition formula 
	\[ F \land x' = 0 \land z' = 0 \equiv x = z = x' = z' = 0 \land y' = y - 1 \land y \geq 0 \] is zero-stable.
\end{mexample}

\subsection{Complete Polynomial Ranking Function Synthesis}
\label{sec:prf-synthesis}

Assuming that a transition formula is zero-stable allows us to ignore
polynomials in the ideal $\qideal{Z} = \set{p \in \mathbb{Q}[X] : F \models p =
0}$ when synthesizing polynomial ranking functions, in the following sense. Suppose there exists $r \in \PRF(F)$ a
polynomial ranking function for $F$, where $\tuple{Z,P} = \algorithmfont{consequence}(F, X)$. We can write $r$
as $r = z+p$ with $z \in \qideal{Z}$ and $p \in \cone{P}$.  Since  $F$ is
zero-stable, we have $F \models p' \leq p - 1$, thus some polynomial in $\cone{P}$ is decreasing.  Thus, it is sufficient to search for decreasing
polynomials in $\cone{P}$.
\autoref{alg:prf} computes the complete set of $\PRF$ for zero-stable transition formulas, which is illustrated in the example below.

\begin{algorithm}
	\caption{Computing $ \PRF $ for zero-stable transition formulas. 
	We use notation
		$ p[y \mapsto z_y: y \in S] $ to denote substitution of all variables
		$ y \in S $ with $ z_y $ in a polynomial $ p $.
		\label{alg:prf}}
	\SetKwFunction{Fconseq}{consequence}
	\SetKwFunction{Finvhom}{inverse-hom}
	\SetKwFunction{Flin}{linearize}
	\SetKwFunction{Fintersect}{intersect}
	\SetKwFunction{Fzstableres}{zero-stable-restrict}
	\SetKwFunction{Fprf}{prf-zero-stable}
	\Fn{\Fprf{$F$}}{
		\KwIn{A zero-stable transition formula $F(X, X')$. }
		\KwOut{A tuple $ Z, R, V $ such that $ \algpoly(Z, R, V) = \PRF(F)$ of $F$. }
		$\tuple{Z, P} \gets \texttt{consequence}(F,X)$\;
		$ Y \gets \set{y_p : p \in P}$ be a set of fresh variables\;
		$ f \gets$ the homomorphism $\mathbb{Q}[Y] \rightarrow \mathbb{Q}[X]$ defined by $f(y_p) = p - p'$\;
		$\tuple{R', V'} \gets \Flin{\Finvhom{\Fconseq{$F$, $X\cup X'$}, $f$, $Y$}}$\;
  \tcc{$\tuple{R_L,V_L}$ represents the polyhedron of linear terms with positive coefficients for variables and constant coefficient $-1$.}
		$ R_L \gets \set{y : y \in Y}, V_L \gets \set{-1}$\; 
		$\tuple{R_Y, V_Y} \gets$ \Fintersect{$R'$, $V'$, $R_L$, $V_L$}\;
  \tcc{Translate polyhedron from $\mathbb{Q}[Y]^1$ back to $\mathbb{Q}[X]$, and add constant $1$.}
		$ R \gets \set{ r[y_p \mapsto p]_{p \in P} : r \in R_Y} $\;
		$ V \gets \set{ 1 + v[y_p \mapsto p]_{p \in P} : v \in V_Y } $\;
		\Return{$ \tuple{Z, R, V} $}
	}
\end{algorithm}

\begin{mexample}
	Consider running \autoref{alg:prf} on a (zero-stable) transition formula
	\begin{align*}
		F: &\quad\, nx \geq 0 \land n \geq 0 \land n' = n \land x \geq 0 \land z \geq 1\\ 
		&\land ((z' = z - 1 \land x' = x ) \lor (x' = x - 1 \land z' = z + n - 1) ) \ .   
	\end{align*}
 The bounded polynomials of $F$ (Line 1) is the algebraic cone defined by 
  $Z = \set{\emptyset}$ and $P = \set{nx, x,
	n, z -1, 1}$. Let $Y = \set{t_{nx}, t_x, t_n, t_{z-1}, t_1}$ be a fresh set of variables, let $f$ be the ring homomorphism such that
	\begin{align*}
		f(t_{nx}) &= nx - n'x' \\
		f(t_x) &= x - x' \\
		f(t_n) &= n - n' \\
		f(t_{z-1}) &= (z-1) - (z'-1) = z - z' \\
            f(t_1) &= 1 - 1 = 0
	\end{align*}
After Line 5, we obtain the polyhedron of linear polynomials in the inverse image of the nonnegative cone of $F$ under $f$, which is
 defined by the rays $R' = \set{t_n,-t_n,t_1,-t_1, t_{nx}+t_{z-1}-1, 1-t_{nx}-t_{z-1}}$ and one vertex $V'=\set{0}$.
 The subset of $\textit{V-rep}(R',V')$ of polynomials with nonnegative coeffients for variables and constant coefficient -1 (Line 7), is the polyhedron defined by rays $R_Y = \set{t_{nx}, t_x, t_n, t_1}$, and vertices $V_Y = \set{t_{nx} + t_{z-1} - 1}$.   Finally, the algorithm returns the algebraic polyhedron 
 with zeros $Z = \emptyset$, positives
 $R = \set{nx, x, n, 1}$, and vertices
 $V = \set{nx + z}$.  Thus $F$ has a non-empty set of polynomial ranking function modulo $\ThZ$ (e.g., it contains $nx+z$), and so we may conclude that $F$ terminates over the reals.
\end{mexample}

We are then ready to prove the correctness of \autoref{alg:prf}.
\begin{theorem}[Soundness and completeness of \autoref{alg:prf}]
	\label{thm:prf-sound-complete}
	For any zero-stable transition formula $ F $
	\[
		\PRF(F) = \algpoly(\algorithmfont{prf-zero-inv}(F))\ .
	\]
\end{theorem}
\begin{proof}
	Let $Z, P, Y, f, R', V', R_L, V_L, R_Y, V_Y, R, V $ be as
	in \autoref{alg:prf}.
        Let
        $s: \mathbb{Q}[Y] \rightarrow \mathbb{Q}[X]$ be the homomorphism
        mapping $y_p \mapsto p$ (corresponding to the substitution on lines 8-9).
        Observe that for any
        linear combination of the $Y$ variables $q = \sum_{p \in P} a_p y_p$, we have
        \[f(q) = \sum_{p \in P} a_p f(y_p) = \sum_{p \in P} a_p (p - p') = \left(\sum_{p \in P} a_p p\right) - \left(\sum_{p \in P} a_p p'\right) = s(q) - s(q)'\ .\]
        
        By definition (line 5), we have $\textit{V-rep}(R_Y, V_Y) = \set{ q
          \in \mathbb{Q}[Y]^1 : F \models f(q) \geq 0}$ and (line 6)
        $\textit{V-rep}(R_L,V_L) = \cone{Y} + \set{-1}$.  It follows that the
        intersection of these two polyhedra (line 7) is
        \[ \textit{V-rep}(R_Y,V_Y) = \set{ q - 1 : q \in \cone{Y}, F \models s(q) - s(q)' - 1 \geq 0}\ .\]

        Then by the construction of $R$ and $V$ (lines 7-8) we have
        \[ \textit{V-rep}(R,V) = \set{ s(q) : q \in \cone{Y}, F \models s(q) - s(q)' - 1 \geq 0}\ .\]
        Letting $K = \textit{V-rep}(R,V)$, we must show that $\PRF(F) = \qideal{Z} + K$.  We prove inclusion in both directions.
        \begin{itemize}
        \item[$\subseteq$] Let $r \in \PRF(F)$.  Then we have $F \models r
          \geq 0$ and $F \models r - r' - 1 \geq 0$.  Since $F \models r \geq
          0$ and $\regcone(Z,P) = \CnX{F}{X}$, we have we have $r = z + p$ for
          some $z \in \qideal{Z}$ and $p \in \cone{P}$.  To show $r = z + p
          \in \qideal{Z} + K$, it is sufficient to show $p \in K$.

          Write $p$ as $\left(\sum_{t \in P} c_t t\right)$ for some
          $\set{c_t}_{t \in P} \subseteq \Qplus$.  Let $q = \left(\sum_{t \in
            P} c_t y_t\right)$.  Then we have $s(q) = p$ and $q \in \cone{Y}$,
          so it is sufficient to show that $F \models s(q) - s(q') - 1 \geq
          0$, or equivalently $F \models p - p' - 1$.  Since $F \models z = 0$
          and $F$ is zero-invariant, we have $z' = 0$.  Since $F \models r -
          r' - 1 \geq 0$ by assumption, we have $F \models (z + p) - (z' + p') - 1 \geq 0$
          and so $F \models p - p' - 1 \geq 0$.
        \item[$\supseteq$] Let $r \in \qideal{Z} + K$.  Then we may write $r$
          as $z+k$ for some $z \in \qideal{Z}$ and $k \in K$.  By the
          definition of $K$, we have $k = s(q)$ for some $q \in \cone{Y}$ such
          that $F \models s(q) - s(q)' - 1 \geq 0 \equiv k - k' - 1 \geq 0$.
          Since $q \in \cone{Y}$ we have $k = s(q) \in \cone{P}$ and so $F
          \models k \geq 0$ and thus $F \models z + k \geq 0$, so $r$ is
          bounded.  Since $F \models k - k' - 1 \geq 0$, $F \models z = 0$,
          and $F$ is zero-stable, we have $F \models (k+z) - (k'-z') - 1$, so
          $r$ is decreasing.  Since $r$ is bounded and decreasing, $r \in \PRF(F)$. \qedhere
        \end{itemize}
\end{proof}

\begin{mexample} \label{eg:incompleteness-prf-due-to-lirr}
Even though \autoref{alg:prf} is complete for synthesizing $\PRF$s modulo $\ThZ$, it does not find all 
$\PRF$s with respect to the standard model. Consider
\[
F \defeq x \geq 1 \land y \geq 1 \land ((x' = 2x \land y' = y/2 - 1) \lor (x' = x/2 - 1 \land y' = 2y))
\]
which admits the $\PRF$ $xy$ since $F \models_{\mathbb{R}} xy \geq 1 \land x'y' \leq xy - 1$.
However the algorithm will not find this $\PRF$ since we cannot derive $F \models_{\ThZ} xy \geq 1$ 
due to the fact that $\ThZ$ lacks axioms governing
the relationship between multiplication and the order relation \cite{kincaidWhenLessMore2023}.
\end{mexample}

\subsection{Proving Termination through Polynomial Ranking Functions}

This section shows how to combine the previous two subsections into a end-to-end termination analysis, which is (1) complete in the sense that it succeeds whenever the input formula has a polynomial ranking function, and (2) monotone in the sense that if $F \models G$ and the analysis finds a termination argument for $G$, then it can also find one for $F$.

Our analysis is presented in \autoref{alg:terminate-prf}, which operates by
first
computing a zero-stable formula and then invoking \autoref{alg:prf} to check if
it has at least one polynomial ranking function.

\begin{algorithm}
	\caption{Proving termination through zero-stable restriction
		and $ \PRF $ synthesis. \label{alg:terminate-prf}}
	\SetKwFunction{Ftermprf}{terminate-PRF}
	\SetKwFunction{Fzstableres}{zero-stable-restrict}
	\SetKwFunction{Fprf}{prf-zero-stable}
	\Fn{\Ftermprf{$F$}}{
		\KwIn{An $ \ThZ $ transition formula $ F $.}
		\KwOut{Whether $ F $ admits a $\PRF$. }
		$ \hat{F} = \Fzstableres{$F$} $\;
		$ \_, \_, V = \Fprf{$\hat{F}$} $\;
		\eIf{$ V = \emptyset $ }{
			\Return{unknown}
		}{
			\Return{true}
		}
	}
\end{algorithm}

\begin{theorem}[Completeness] \label{thm:prf-syn-completeness}
	If $ F $ has a polynomial ranking function (modulo $\ThZ$), then \autoref{alg:terminate-prf}
	returns $\true$ on $ F $.
\end{theorem}
\begin{proof}
	Suppose $ F $ has $ r $ as a $ \PRF $. Since $\hat{F}=\algorithmfont{zero-stable-restrict}(F)$ entails $F$ (\autoref{lem:zero-stable-restrict-properties}),
	$ r $ is also a $ \PRF $ of $ \hat{F} $. 
 Letting $\tuple{Z, P, V} = \algorithmfont{prf-zero-inv}(\hat{F}) $, we have $r \in \algpoly(Z,P,V)$
 by \autoref{thm:prf-sound-complete}, and so $V$ is non-empty, and \autoref{alg:terminate-prf} returns true.
\end{proof}

\begin{mexample}
    The reverse of \autoref{thm:prf-syn-completeness} does not hold. Due to zero-stable restriction, \autoref{alg:prf} can even prove termination of loops that 
    do not admit $\PRF$s even in the standard model.
    For example, it can prove termination of $F \defeq x = 0 \land x' \neq 0$ since its zero-stable restriction is
	unsatisfiable. 
    To see that $F$ does not admit any $\PRF$, suppose for a contradiction that it has $r$ as a $\PRF$.
    But this is impossible since there exists $x'$ such that $r(x') > r(0) - 1$ due to the continuity of $r$.
\end{mexample}

The completeness of the ranking function synthesis procedures leads to several
desirable properties of behavior of the resulting termination analysis, one of
which is \emph{monotonicity}, i.e., if the analysis succeeds on a transition
formula $G$, then it is guaranteed to succeed on a stronger one $F$. Further,
monotone termination analysis on loops can be lifted to monotone whole-program
analysis by the framework presented by Zhu et al.
\cite{zhuTerminationAnalysisTears2021}.

\begin{theoremEnd}{corollary}[Monotonicity]
	If $F \models G$ and \algorithmfont{terminate-PRF}($G$) returns $\true$, then \algorithmfont{terminate-PRF}($F$) returns $\true$.
\end{theoremEnd}
\begin{mproofEnd}
	Let 
	\begin{align*}
		\hat{F} &= \algorithmfont{zero-stable-restrict}(F) \\
		G_0 &= \algorithmfont{zero-stable-restrict}(G) \ .
	\end{align*}
	 Since $\hat{F} \models F$ and $F \models G$, we have $\hat{F} \models G$.  Since $\hat{F} \models G$, $\hat{F}$ is zero-stable, and $G_0$ is the weakest zero-stable formula that entails $G$, we have $\hat{F} \models G_0$. Since 
	 \algorithmfont{terminate-PRF}($G$) = true,  $\PRF(G_0)$ is non-empty, and so $\PRF(\hat{F})$ is non-empty, and thus \algorithmfont{terminate-PRF}($F$) = true.
\end{mproofEnd}



\section{Lexicographic Polynomial Ranking for Integer Transitions}
\label{sec:lprf}

In this section, we show how to synthesize lexicographic polynomial ranking
functions. 
The strategy (inspired by \cite{ben-amramRankingFunctionsLinearConstraint2014}) is based on the connection between $\WLPRF$ and 
quasi-ranking functions. 
We can describe the set of quasi-ranking functions as the intersection of
the sets of bounded and non-increasing polynomials of $F$:
\[ \QPRF(F) = \textit{Bounded}(F) \cap \textit{Noninc}(F) \]
where
\begin{align*}
    \textit{Bounded}(F) &\defeq \set{ p \in \mathbb{Q}[X] : F \models p \geq 0}\\
    \textit{Noninc}(F) &\defeq \set { p \in \mathbb{Q}[X] : F \models p \geq p' } \ .
\end{align*}
In the following, we first show how to synthesize 
$\QPRF$s (Section~\ref{sec:qprf}), using which we are able to synthesize $\WLPRF$s (Section~\ref{sec:plrf-for-termination}) to prove termination.
Similar to Section~\ref{sec:prf}, we need to compute 
zero-stable restriction of transition formulas 
to make sure that the set of ranking arguments found is complete.

\subsection{Synthesizing Polynomial Quasi-Ranking Functions}
\label{sec:qprf}

\autoref{alg:qprf} finds all $ \QPRF $s for a zero-stable transition formula $ F $, using a variation of our strategy for finding $\PRF$s.

\begin{algorithm}
	\caption{Computing $ \QPRF $ for zero-stable transitions. \label{alg:qprf}}
	\SetKwFunction{Fconseq}{consequence}
	\SetKwFunction{Ftermprf}{terminates-WPRF}
	\SetKwFunction{Fzstableres}{zero-inv-restrict}
	\SetKwFunction{Fqprf}{qprf-zero-stable}
	\SetKwFunction{Finvhom}{inverse-hom}
	\SetKwFunction{Flin}{linearize}

	\Fn{\Fqprf{$F$}}{
		\KwIn{A zero-stable transition formula $F \in \mathbf{TF}(X)$. }
		\KwOut{The algebraic cone of all $\QPRF$s of $F$. }
		$\tuple{Z, P} \gets \Fconseq{$F,X$}$\;
		$ Y \gets \set{y_p : p \in P}$ be a set of fresh variables\;
		$ f \gets$ the ring homomorphism $\mathbb{Q}[Y] \rightarrow \mathbb{Q}[X]$ defined by $f(y_p) = p - p'$\;
            \tcc{For any $t \in \cone{R}$, we have $F \models f(t) \geq 0$. }
		$\tuple{R', V'} \gets \Flin{\Finvhom{\Fconseq{$F$, $X\cup X'$}, $f$, $Y$}}$\;
            \tcc{$\tuple{R_L,V_L}$ represents the polyhedron of linear terms with positive coefficients for variables and constant coefficient $0$.}
		$ R_L \gets \set{y : y \in Y}, V_L \gets \set{0}$\; 

            \tcc{The intersection of $\textit{V-rep}(R',V')$ and $\textit{V-rep}(R_L,V_L)$ is a cone.}
            $\tuple{R,\_} \gets $ \Fintersect{$R'$, $V'$, $R_L$, $V_L$}\;
  	    $P_X \gets \set{ r[y_p \mapsto p: p \in P]: r \in R}$\;
		\Return{$\tuple{Z,P_X}$}
	}
\end{algorithm}

\begin{theorem}[Soundness and completeness of \autoref{alg:qprf}]
	\label{thm:qprf-sound-complete}
	Suppose $ F $ is a zero-stable transition formula. 
	Then 
	\[ 
		\QPRF(F) = \regcone(\algorithmfont{qprf-zero-stable}(F)) \ .
	\]
\end{theorem}
\begin{proof}
  Let $Z, P, Y, f, R, V, P_X$ be as in \autoref{alg:qprf}.
  Let $s: \mathbb{Q}[Y] \rightarrow \mathbb{Q}[X]$ be the homomorphism
  mapping $y_p \mapsto p$ (corresponding to the substitution on line 6),
  and observe that for any
  linear combination of the $Y$ variables $q = \sum_{p \in P} a_p y_p$, we have
  $f(q) = s(q) - s(q)'$ (as in Theorem~\ref{thm:prf-sound-complete}).
  By construction (lines 5-8) we have
  \[\cone{P_X} = \set{ s(q) : q \in \cone{Y}, F \models f(q) \geq 0 }\ ,\]
  which by the above observation can be written equivalently as
  \[\cone{P_X} = \set{ s(q) : q \in \cone{Y}, F \models s(q) - s(q)' \geq 0 }\ .\]
  Since $\set{ s(q) : q \in \cone{Y} }$ is precisely $\cone{P}$, we have 
  \[\cone{P_X} = \set{ p \in \cone{P}: F \models p - p' \geq 0 } \]

  We must show that $\QPRF(F) = \qideal{Z} + \cone{P_X}$.  We prove inclusion in both directions.
  \begin{itemize}
  \item[$\subseteq$] Let $r \in \QPRF(F)$.  Since $F \models r \geq 0$ and $\regcone(Z,P) = \CnX{F}{X}$, we must
    have $r = z + p$ for some $z \in \qideal{Z}$ and $p \in \cone{P}$.  It is
    sufficient to show that $p \in \cone{P_X}$.  Since $F$ is zero-stable and
    $F \models z = 0$, we have $F \models z - z' = 0$ and so we must have $F
    \models p - p' \geq 0$.  It follows from the above that $p \in \cone{P_X}$.
    
  \item[$\supseteq$] Since $\QPRF(F)$ is a cone it is closed under addition,
    so it is sufficient to prove that $\qideal{Z} \subseteq \QPRF(F)$ and
    $\cone{P_X} \subseteq \QPRF(F)$.  Since $F$ is zero-stable, we have
    $\qideal{Z} = \set{ z \in \mathbb{Q}[X] : F \models z = 0} \subseteq
    \QPRF(F)$.  Since $\cone{P_X} = \set{ p \in \cone{P}: F \models p - p'
      \geq 0 }$, we have that each $p \in \cone{P}$ is both bounded ($p \in
    \cone{P}$) and non-increasing ($F \models p - p' \geq 0$), and thus
    belongs to $\QPRF(F)$.  \qedhere
  \end{itemize}
\end{proof}

\subsection{Lexicographic Polynomial Ranking Functions}
\label{sec:plrf-for-termination}

Given \autoref{alg:zero-stable-restrict} for computing zero-stable restrictions
and \autoref{alg:qprf} for finding $ \QPRF $s, we present \autoref{alg:lprf} for
proving termination by finding $\WLPRF$s.

\begin{algorithm}
	\caption{Proving termination by synthesizing lexicographic polynomial ranking functions. \label{alg:lprf}}
	\SetKwFunction{Fsat}{sat-lirr}
	\SetKwFunction{Fconseq}{consequence}
	\SetKwFunction{Ftermprf}{terminates-WPRF}
	\SetKwFunction{Fzstableres}{zero-stable-restrict}
	\SetKwFunction{Fzstableqprf}{qprf-zero-stable}
	\SetKwFunction{Flprf}{terminate-lprf}
	\Fn{\Flprf{$F$}}{
		$Z \gets \emptyset$\;
		\Repeat{$\qideal{Z} = \qideal{Z'}$}{
			$Z' \gets Z$\;
			$\tuple{Z, P} \gets \Fzstableqprf{\Fzstableres{$F$}}$\;
			$F \gets F \land \bigwedge_{z \in Z} z' = z \land \bigwedge_{p \in P} p' = p$\;
		}
		\eIf{$1 \in \qideal{Z}$}{
			\Return{true}
			\tcc{$F$ is unsatisfiable modulo $\ThZ$ iff $1 \in \qideal{Z}$}
		}{
			\Return{unknown}
		}
	}
	
\end{algorithm}

Ignoring the effects of zero-stable restriction, 
Algorithm~\ref{alg:lprf} iteratively computes a sequence of  algebraic cones that represent
all $\QPRF$s, and finally checks if all transitions in $F$ have been ranked.

\begin{mexample}
	Consider the transition formula
	\[ F: x - xy \geq 0 \land y \geq 0 \land ((x' = x \land y' = y - 1) \lor (y \geq 1 \land x' = x - 1 \land y' = y)) \]
	(which has a dimension-2 $\WLPRF$ $\tuple{y, x-xy}$).
		   The following table depicts the execution of \autoref{alg:lprf}, displaying a (simplified) transition formula $F$, zero polynomials $Z$, and positive polynomials $P$ after each
		iteration of the loop,
	 culminating in $F=\false$, which indicates that $F$ terminates.

		\vspace*{5pt}
		\begingroup
		{\footnotesize
		\centering
		\setlength{\tabcolsep}{15pt}
		\begin{tabular}{m{2cm}m{4cm}m{4cm}m{4cm}}
			\toprule
			\multicolumn{1}{c}{}  & \multicolumn{1}{c}{$F$} & \multicolumn{1}{c}{$Z$} & \multicolumn{1}{c}{$P$} \\
			\midrule
			\multicolumn{1}{c}{Before} & \multicolumn{1}{c}{$\begin{aligned} 
				&\quad\, x - xy \geq 0 \land y \geq 0 \\ 
				&\land ((x' = x \land y' = y - 1) \\
				&\quad\, \lor (y \geq 1 \land x' = x - 1 \land y' = y)) 
			  \end{aligned}$}
			&  \multicolumn{1}{c}{$\emptyset$} & \multicolumn{1}{c}{-} \\
			\midrule
			\multicolumn{1}{c}{Iter 1} & \multicolumn{1}{c}{$\begin{aligned} 
				&\quad\, x - xy \geq 0 \land y \geq 0 \\ 
				&\land (y\geq 1 \land x' = x - 1 \land y' = y)) 
			  \end{aligned}$}
			&  \multicolumn{1}{c}{$\emptyset$} & \multicolumn{1}{c}{$\set{y}$} \\
			\midrule
			\multicolumn{1}{c}{Iter 2} & \multicolumn{1}{c}{$\false$}
			&  \multicolumn{1}{c}{$\emptyset$} & \multicolumn{1}{c}{$\set{y, x-xy}$} \\
			\bottomrule
		\end{tabular}\par
		}
		\endgroup
	\end{mexample}

\begin{theorem}[Correctness of \autoref{alg:lprf}]
	\label{thm:lprf-correctness}
	 \autoref{alg:lprf} is a terminating procedure, and for any transition formula $F$ for which $\algorithmfont{terminate-lprf}(F) = \true$, we have that $F$ terminates
	over the integers.
\end{theorem}
\begin{proof}
	Let $ F^{(k)}, Z^{(k)}, P^{(k)} $ denote the values of $ F $, $ Z $, and $P$ at the
	beginning of $ k $-th iteration of the loop in \autoref{alg:lprf}. We first
	prove termination of the algorithm. Suppose the loop does not terminate,
	then $\qideal{Z^{(k+1)}} \supsetneq \qideal{Z^{(k)}}$ for all iterations
	$k$. We have thus obtained an infinite and strictly ascending chain of
	ideals in the polynomial ring $ \mathbb{Q}[X \cup X'] $, contradicting Hilbert's basis theorem.

	Now we show that if \autoref{alg:lprf} returns $\true$, all integer
 executions of $F$ terminate. We prove this by induction on $n$, the number of
 times the loop runs in Algorithm~\ref{alg:lprf}. The base case holds when $n =
 1$ since the zero-stable restriction of $F$ being unsatisfiable modulo $\ThZ$
 implies that $F$ terminates. Suppose that the proposition is true for $n \geq
 1$ and we want to prove the case of $(n+1)$. Consider the first iteration of
 the loop in Algorithm~\ref{alg:lprf}. For convenience, we use $F$ to denote
 $F^{(1)}$, $\hat{F}$ to denote the zero-stable restriction of $F$, and $F'$ to
 denote $F \land \bigwedge_{z \in Z^{(1)}} z' = z \land \bigwedge_{p \in
 P^{(1)}} p' = p$. By the inductive hypothesis, $F'$ terminates. 
 Suppose for a contradiction that $F$ does not terminate.
 By \autoref{lem:zero-stable-restrict-properties} we know $\hat{F}$ also does 
 not terminate. 
 Define $r = \sum_{p \in P^{(1)}} p$, then $r \in \QPRF(\hat{F})$ due to
 \autoref{thm:qprf-sound-complete}. By \autoref{lem:quasi-rf-implies-termination},
 $\hat{F} \land r' = r$ has an infinite integer execution $x_0, x_1, \dots$ since 
 $\hat{F}$ has one.
 Let $i \in \mathbb{N}$ be arbitrary.
 Since $x_i \rightarrow_{\hat{F} \land r' = r} x_{i+1}$,
 we know that $ \sum_{p \in P^{(1)}} p(x_{i+1}) - p(x_i) =
 0$. This is a sum of nonpositive terms because $p(x_{i+1}) \leq p(x_i)$ holds
 for any $p \in P^{(1)}$ due to $p \in
 \QPRF(\hat{F})$. Thus for all $p \in P^{(1)}$, it holds that $p(x_{i+1}) = p(x_i)$.
 Since $\hat{F}$ is zero-stable, we have that $z(x_{i+1}) = z(x_i) = 0$ for all $z \in
 Z^{(1)}$. Thus we have $x_i \rightarrow_{F'} x_{i+1}$
 and subsequently  $x_0, x_1, \dots$ is an infinite integer execution of $F'$,
 contradicting the inductive hypothesis that $F'$ terminates. 
\end{proof}

The following theorem states that even though we operate modulo $\ThZ$, we have a guarantee on the capability
of the ranking functions synthesized that it is no less powerful
than LLRF modulo the standard linear integer arithmetic, under mild assumptions.

\begin{theorem}[Subsumption of LLRFs]
	\label{thm:lprf-subsumes-llrf}
	If $F \in \mathbf{TF}(X)$ is a negation-free formula involving
        only linear polynomials and $F$ has an LLRF modulo linear integer arithmetic ($\LIA$), 
        then $F \land \bigwedge_{x \in (X \cup X')} \Int(x)$ has a $\WLPRF$ modulo $\ThZ$.
\end{theorem}
\begin{proof}
	We first prove a lemma as follows. Let $F(Y)$ be a ground, negation-free,
	$\LIA$ transition formula over variable set $Y$. Then for any affine term
	$r$ over $Y$, if $F \models_{\LIA} r \geq 0$ then $F \land \bigwedge_{y \in
	Y} \Int(y) \lirrmodels r \geq 0$. (The proof is similar to Theorem 8 in
	\cite{kincaidWhenLessMore2023}. Without loss of generality we assume $F$ is
	a conjunctive formula. Suppose $F \models_{\LIA} r \geq 0$. By \cite{schrijverCuttingPlanesResearch1980,chvatalEdmondsPolytopesHierarchy2006} there is
	a cutting-plane proof of $r \geq 0$ from $F$. Since each inference rule in a
	cutting-plane proof is valid $\ThZ$, we have that  $F \land \bigwedge_{y \in Y} \Int(y)
	\lirrmodels r \geq 0$.)
	
	Suppose that $\LIA$ formula $F$ admits an LLRF $r_1, \dots, r_n$ of
	dimension $n$, then $F \models_{\LIA} r_i \geq 0$ for each $i$ (bounded), $F
	\land \bigwedge_{j=1}^{i-1} r_j' = r_j \models_{\LIA} r_i' \leq r_i$ for
	each $i$ (decreasing), and $F \land \bigwedge_{j=1}^n r_j' = r_j
	\models_{\LIA} \false $ (coverage). Since the left hand side of all the
	implications listed above contains ground linear formulas without negation
	and the right hand side all contains linear inequalities (with $\false$
	being interpreted as $0 \leq -1$), these implications also hold modulo
	$\ThZ$ by the lemma. Therefore, $r_1, \dots, r_n$ is a $\WLPRF$ of $F$.
 \end{proof}

 \autoref{alg:lprf} is also complete w.r.t. the existence of $\WLPRF$s, since it finds
 a $\WLPRF$ if there exists one for the transition formula. Moreover, it is optimal in terms 
 of the dimension of the $\WLPRF$ found.
\begin{theorem}[Completeness of \autoref{alg:lprf} w.r.t. $\WLPRF$]
	\label{thm:lprf-completeness}
        If a transition formula $ F $ admits a $ \WLPRF $ of dimension $N$,
	then  $\algorithmfont{termination-lprf}(F)$ returns $\true$ and \autoref{alg:lprf}
        terminates in no more than $N$ iterations.
\end{theorem}
\begin{proof}
	Suppose that $ r_1, \dots, r_N $ is a $\WLPRF$ for $F$.  Let $F^{(k)}$ denote the value of $F$ after the $k$th iteration of the while loop in \autoref{alg:lprf}, with the convention that if the loop exits after $m$ iterations then $F^{(m)} = F^{(m+1)} = \dotsi$.  For any $k$, let
 $\tuple{Z^{(k)},P^{(k)}} \defeq \Fzstableqprf(\Fzstableres(F^{(k)}))$.
 
 We prove that $r_i \in \regcone(Z^{(i)},P^{(i)})$ for all $i$, by induction on $i$.
 For the base case, $r_1$ is a quasi ranking function for $F$ and so also a quasi ranking function for the zero-stable restriction $F$, and thus $r_1 \in \regcone(Z^{(1)},P^{(1)})$.
 For the inductive step, we have $r_j \in\regcone(Z^{(j)},P^{(j)})$ for all $j \leq i$, and we must prove  $r_{i+1} \in\regcone(Z^{(i+1)},P^{(i+1)})$.  By the inductive hypothesis, we have $r_1,\dots,r_i \in \regcone(Z^{(i)},P^{(i)})$.  It follows that
 $F^{(i+1)} \models F^{(i)} \land \bigwedge_{j=1}^i r_j' - r_j$, so by the (Decreasing) condition of $\WLPRF$, $r_{i+1}$ is a quasi ranking function of $F^{(i)}$.  It follows that $r_{i+1}$ is a 
 quasi ranking function of 
 $\Fzstableres(F^{(i)})$, and thus $r_{i+1}$ belongs to 
 $\regcone(Z^{(i+1)},P^{(i+1)})$.

 By the (Coverage) condition of $\WLPRF$, we have
 that $F \land \bigwedge_{j=1}^N r_j' = r_j$ is unsatisfiable.  Since for each $j$ we have
 \[ r_j \in \regcone(Z^{(j)},P^{(j)}) \subseteq \regcone(Z^{(N)},P^{(N)}) \] we must have that
 $F^{(N)}$ is unsatisfiable, and so $F^{(N)} \models 1 = 0$, and thus $\algorithmfont{termination-lprf}(F)$ returns $\true$.
\end{proof}

\begin{theoremEnd}{corollary}[Monotonicity of \autoref{alg:lprf}]
	\label{thm:lprf-monotone}
 Let $F$ and $G$ be transition formulas with $ F \models G $. If  $
	\algorithmfont{termination-lprf}(G) = \true $, then it is guaranteed that $
	\algorithmfont{termination-lprf}(F) = \true $.
\end{theoremEnd}
\begin{mproofEnd}
	We prove by induction that at the start of each iteration in
	\autoref{alg:lprf}, we always have $F^{(k)} \models G^{(k)}$. The base case
	$ F^{(0)} = F \models G = G^{(0)}$ is trivial. Now assume the statement is
	true for $F^{(k)} \models G^{(k)}$. Then their zero-stable restrictions $\hat{F}
	\models G_0$ by \autoref{lem:zero-stable-restrict-properties}. 
	Thus $\QPRF(\hat{F}) \supseteq \QPRF(G_0)$. Now suppose that
	\begin{align*}
		F^{(k+1)} &= F^{(k)} \land \bigwedge_{z_i \in Z} z_i' = z_i \land
		\bigwedge_{p_j \in P} p_j' = p_j \\
		G^{(k+1)} &= G^{(k)} \land \bigwedge_{n_k \in N} n_k' = n_k \land
		\bigwedge_{m_l \in M} m_l' = m_l
	\end{align*}
	For each $n_k$ or $m_l$, it is contained in the set $\QPRF(G_0)$ and therefore 
	is also contained in the set $\QPRF(\hat{F})$. Since all $\QPRF$ of $\hat{F}$ can be 
	written as $\sum_{z_i \in Z} f_i z_i + \sum_{p_j \in P} \lambda_j p_j, 
	f_i \in \mathbb{Q}[X], \lambda_j \in \Qplus$, so does any $n_k$ or $m_l$,
	Thus $\bigwedge_{z_i \in Z} z_i' = z_i \land
	\bigwedge_{p_j \in P} p_j' = p_j \models n_k' = n_k \land m_l' = m_l$ for any 
	$n_k$ and $m_l$. And consequently $F^{(k+1)} \models G^{(k+1)}$.
	Given this, if \autoref{alg:lprf} returns $\true$ for $G$, then for 
	some $G^{(t)}$ it is unsatisfiable modulo $\ThZ$. But $F^{(t)} \models G^{(t)}$
	so it must also be unsatisfiable modulo $\ThZ$, and the algorithm will 
	also return $\true$ on input $F$.
\end{mproofEnd}

\section{Evaluation} 
\label{sec:evaluation}

We consider two key research questions in the experimental evaluation.
First, how does the proposed technique perform in proving termination
of linear or nonlinear programs comparing to existing tools, in terms
of running time and the number of tasks solved.
We thus compare the proposed techniques with other \emph{sound} and \emph{static}
provers for termination. In particular, we compare against Ultimate
Automizer
\cite{heizmannUltimateAutomizerCommuHash2023,heizmannRefinementTraceAbstraction2009}
and 2LS \cite{chenBitPreciseProcedureModularTermination2018}, which are the top
two \emph{sound} tools in the \texttt{Termination} category in the 12th
Competition on software verification (SV-COMP 2023).
We also report a
qualitative comparison with the dynamic tool DynamiTe
\cite{leDynamiTeDynamicTermination2020}.
Second, we have shown in \autoref{thm:lprf-subsumes-llrf} that LPRF
subsumes LLRF synthesis for proving termination under certain assumptions, 
but we would like to understand the performance overhead of our more
general procedure.
We  compare with the
LLRF synthesis procedure implemented in ComPACT
\cite{zhuTerminationAnalysisTears2021}.

\paragraph{Implementation} We implement polynomial ranking functions synthesis (Section~\ref{sec:prf}) and
 lexicographic polynomial ranking function synthesis (Section~\ref{sec:lprf}) as
two mortal precondition operators (i.e., an operator that takes in a transition formula
representing a single loop and outputs sufficient terminating conditions for that loop)
in the ComPACT termination analysis framework
\cite{zhuTerminationAnalysisTears2021}, also utilizing the $\ThZ$ solver,
consequence finding, inverse homomorphism, and nonlinear invariant generation
procedures from Kincaid et al. \cite{kincaidWhenLessMore2023}. Given any loop,
we first try synthesizing polynomial ranking functions, and only attempt to
synthesize lexicographic polynomial ranking function upon failure. Our
implementation is denoted by ``LPRF'' in the tables. We have also combined our technique 
with phase analysis, a technique for improving termination analyzers by 
analyzing phase transition structures implemented in 
ComPACT \cite{zhuTerminationAnalysisTears2021}.

\paragraph*{Environment.}
We ran all experiments in a virtual machine with Lubuntu 22.04 LTS (kernel
version 5.12) with a CPU of Intel Core i7-9750H @ 2.60 GHz and 8 GB of memory.
The SV-COMP 2023 binaries of Ultimate Automizer v0.2.2-2329fc70 and 2LS version
0.9.6-svcomp23 are used in the experiments. All tools were run under a time
limit of 2 minutes.

\paragraph*{Benchmarks.}
We collected tasks from the SV-COMP 2023 \texttt{Termination} benchmarks.
Since the focus of the proposed technique is to prove termination of nonlinear
programs, we divide the tasks into two suites according to whether they require
nonlinear reasoning. The \texttt{linear} suite consists of \emph{terminating}
and \emph{nonrecursive} integer programs from the
\texttt{Termination-MainControlFlow} subcategory in the SV-COMP, excluding the \texttt{termination-nla} folder. The
\texttt{nonlinear} suite contains terminating programs without
overflow\footnote{Our technique assumes unbounded integers but 2LS is
    bit-precise and requires this constraint.} in the \texttt{termination-nla}
folder. This suite was originally presented in
\cite{leDynamiTeDynamicTermination2020} and contains only integer programs.

\begin{figure}[ht]
    \centering
    \begin{minipage}[t]{0.45\textwidth}
        \vspace{-125pt}
        \captionof{table}{Experimental results on termination verification
            benchmarks comparing our technique (LPRF) with lexicographic linear
            ranking function (LLRF) synthesis, both techniques with phase
            analysis (+$\Phi$), as well as 
            ComPACT, Ultimate Automizer, and 2LS. The \#c row counts the number of
            solved tasks, t reports total running time in seconds, excluding timeouts (\# timeouts in parentheses). \label{tab:eval-static}}
        \footnotesize
        \scalebox{0.9}{
            \begin{tabular}{cccc}
                \toprule
                \multicolumn{1}{l}{}         & \multicolumn{1}{l}{}  & \multicolumn{1}{c}{\texttt{linear}}   & \multicolumn{1}{c}{\texttt{nonlinear}} \\
                \multicolumn{2}{c}{\#tasks}  & 171                   & 26                                                                             \\
                \midrule
                \multirow{2}{*}{LPRF}        & \#c                   & 118                                   & \textbf{17}                            \\
                                             & \multicolumn{1}{r}{t} & \multicolumn{1}{r}{333.2 (2)}         & \multicolumn{1}{r}{\textbf{47.8} (0)}  \\
                \midrule
                \multirow{2}{*}{LPRF+$\Phi$} & \#c                   & 132                                   & \textbf{17}                            \\
                                             & \multicolumn{1}{r}{t} & \multicolumn{1}{r}{426.0 (2)}         & \multicolumn{1}{r}{119.5 (0)}          \\
                \midrule
                \multirow{2}{*}{LLRF}        & \#c                   & 120                                   & 3                                      \\
                                             & \multicolumn{1}{r}{t} & \multicolumn{1}{r}{\textbf{74.6} (0)} & \multicolumn{1}{r}{161.1 (1)}          \\
                \midrule
                \multirow{2}{*}{LLRF+$\Phi$} & \#c                   & 138                                   & 4                                      \\
                                             & \multicolumn{1}{r}{t} & \multicolumn{1}{r}{98.6 (0)}          & \multicolumn{1}{r}{263.0 (2)}          \\
                \midrule

                \multirow{2}{*}{ComPACT}     & \#c                   & 140                                   & 4                                      \\
                                             & \multicolumn{1}{r}{t} & \multicolumn{1}{r}{105.7 (0)}         & \multicolumn{1}{r}{288.8 (2)}          \\
                \midrule
                \multirow{2}{*}{UAutomizer}  & \#c                   & \textbf{160}                          & 1                                      \\
                                             & \multicolumn{1}{r}{t} & \multicolumn{1}{r}{2423.1 (6)}        & \multicolumn{1}{r}{1282.7 (8)}         \\
                \midrule
                \multirow{2}{*}{2LS}         & \#c                   & 114                                   & 0                                      \\
                                             & \multicolumn{1}{r}{t} & \multicolumn{1}{r}{5399.1 (43)}       & \multicolumn{1}{r}{2748.7 (20)}        \\
                \bottomrule
            \end{tabular}
        }
    \end{minipage}
    \begin{minipage}[t]{0.45\textwidth}
        \centering
        \resizebox{\columnwidth}{!}{%
            \begin{tikzpicture}
                \begin{axis}[
                        xlabel={Instances Solved},
                        ylabel={Time (s)},
                        grid=major,
                        legend entries={LPRF,LPRF+$\Phi$,LLRF,LLRF+$\Phi$,ComPACT,UAutomizer,2LS},
                        legend pos=north west,
                        mark size=1.5pt,
                        xmax=171,
                    ]
                    \addplot table {figures/linear/compact-prf-no-phase.dat};
                    \addplot table {figures/linear/compact-prf.dat};
                    \addplot table {figures/linear/ComPACT-llrf-only-no-phase.dat};
                    \addplot table {figures/linear/ComPACT-llrf-only-with-phase.dat};
                    \addplot table {figures/linear/ComPACT.dat};
                    \addplot table {figures/linear/2ls.dat};
                    \addplot table {figures/linear/uautomizer.dat};
                \end{axis}
            \end{tikzpicture}
        }
        \captionof{figure}{Linear benchmarks. \label{fig:cactus-lin}}
        \vspace{20pt}

        \resizebox{\columnwidth}{!}{%
            \begin{tikzpicture}
                \begin{axis}[
                        xlabel={Instances Solved},
                        ylabel={Time (s)},
                        grid=major,
                        legend entries={LPRF,LPRF+$\Phi$,LLRF,LLRF+$\Phi$,ComPACT,UAutomizer},
                        legend style={at={(0.6, 0.8)},anchor=north west},
                        mark size=1.5pt,
                        xmax=30,
                    ]
                    \addplot table {figures/nonlinear/compact-prf-no-phase.dat};
                    \addplot table {figures/nonlinear/compact-prf.dat};
                    \addplot table {figures/nonlinear/ComPACT-llrf-only-no-phase.dat};
                    \addplot table {figures/nonlinear/ComPACT-llrf-only-with-phase.dat};
                    \addplot table {figures/nonlinear/ComPACT.dat};
                    \addplot table {figures/nonlinear/uautomizer.dat};
                \end{axis}
            \end{tikzpicture}
        }
        \captionof{figure}{Nonlinear benchmarks. 2LS cannot solve any task in the suite and is thus omitted 
        in the plot. \label{fig:cactus-nonlin}}
    \end{minipage}

\end{figure}

\paragraph*{Comparing against sound and static analyses.}
The results of running all experiments are presented in
Table~\ref{tab:eval-static}. For the \texttt{nonlinear} suite, our proposed
techniques for synthesizing polynomial ranking functions and lexicographic
ranking arguments perform significantly better than the current static analysis
tools in terms of both number of tasks proved and running speed. Our
technique subsumes linear lexicographic ranking function synthesis for a large class of
integer variable programs, and thus remains competitive for the \texttt{linear} suite.
We see that there is a moderate slowdown comparing to linear
lexicographic ranking function synthesis implemented in ComPACT. As a top
competitor in the SV-COMP, Ultimate Automizer proves the most tasks in the
\texttt{linear} suite, while requiring more time to run compared to our
techniques (see the cactus plots Fig~\ref{fig:cactus-lin} and
\ref{fig:cactus-nonlin}).


\paragraph*{Comparing against DynamiTe.} The DynamiTe paper
\cite{leDynamiTeDynamicTermination2020} presents a \emph{dynamic} technique that
can guess and verify linear or quadratic ranking functions for nonlinear
programs and proposes a benchmark suite \texttt{termination-nla}
for termination of nonlinear programs. Due to hardware constraints, we could not reproduce the original
evaluations for DynamiTe in our evaluation environment.
Instead, we perform a comparison
with the results reported in the paper. Since our tool
is automated and sound but can only prove termination, we only count the
terminating programs for which DynamiTe can automatically validate the discovered ranking
functions.
In the \texttt{termination-nla} suite,
DynamiTe can learn the ranking function 
for most tasks ($23$ out of $26$) but can only automatically validate $7$ of them, whereas our static analysis technique LPRF is able to automatically prove $17$.
This observation demonstrates that verifying a
given ranking function modulo nonlinear integer arithmetic
is not only difficult in theory but remains challenging for 
modern arithmetic theory solvers. This provides additional motivation for the introduction of
the weak arithmetic theory $\ThZ$ in this work.



\section{Related Work} 
\label{sec:related_work}

\paragraph*{Ranking function synthesis}
For linear loops, there are \emph{complete} procedures for synthesizing
particular classes of ranking functions such as linear
\cite{podelskiCompleteMethodSynthesis2004,ben-amramRankingFunctionsLinearConstraint2014},
lexicographic linear
\cite{bradleyLinearRankingReachability2005,ben-amramRankingFunctionsLinearConstraint2014},
multi-phase \cite{ben-amramMultiphaseLinearRankingFunctions2019}, and nested
\cite{leikeRankingTemplatesLinear2014}. For nonlinear loops, it is usually
necessary to start with a template, e.g., polyranking functions based on
a finite tree of differences between terms
\cite{bradleyPolyrankingPrinciple2005a}, or limiting the degree of the
polynomial ranking functions to be considered
\cite{carbonneauxCompositionalCertifiedResource2015,leDynamiTeDynamicTermination2020}.
Other procedures for synthesizing (bounded-degree) polynomial ranking functions
rely on semidefinite programming \cite{cousotProvingProgramInvariance2005} and cylindrical
algebraic decomposition \cite{chenDiscoveringNonlinearRanking2007}, but we have not found
implementations for these techniques to compare with experimentally.
Chatterjee et al. \cite{chatterjeeTerminationAnalysisProbabilistic2016} synthesizes polynomial
ranking supermartingales for probabilistic programs through Positivestellensatz,
which bears some resemblance to our approach based on $\ThZ$ consequence
finding. One key advantage of our work comparing to previous work is the
completeness and monotonicity guarantee. 

\paragraph*{Decision procedures for termination}
The decision problem for termination of linear loops was introduced by Tiwari
\cite{tiwariTerminationLinearPrograms2004a}. General procedures for loops over
the reals was developed by Tiwari \cite{tiwariTerminationLinearPrograms2004a},
over the rationals by Braverman \cite{bravermanTerminationIntegerLinear2006},
and over the integers by Hosseini et al.
\cite{hosseiniTerminationLinearLoops2019b}. Time complexity for linear and
lexicographic linear ranking function synthesis has also been studied
\cite{ben-amramRankingFunctionsLinearConstraint2014}. For nonlinear loops, it
has been shown that termination of certain restricted classes of single-path
polynomial loops over the reals are decidable, e.g., when the guard is compact
and connected \cite{liTerminationSinglePathPolynomial2016a}, when the loop is
triangular weakly nonlinear \cite{harkTerminationTriangularPolynomial2022}, when
the guard is compact semi-algebraic and the body contains continuous
semi-algebraic updates \cite{neumannRankingFunctionSynthesis2020}. Additionally,
Neumann et al. \cite{neumannRankingFunctionSynthesis2020} presents a non-constructive method 
for reasoning about termination via polynomial ranking functions of unbounded degree. 
The authors have not found any work that handles polynomial loops over
\emph{integers} without assuming real relaxations.


\subsubsection*{Acknowledgements.}
This work was supported in part by the NSF under grant
number 1942537. Opinions, findings,
conclusions, or recommendations expressed herein are those of the authors and
do not necessarily reflect the views of the sponsoring agencies.

%
%
%
\bibliographystyle{splncs04}
\bibliography{references.bib}

\end{document}